%% file: lossykernel-arxiv.tex
\documentclass[a4paper,UKenglish]{lipics-v2016-arxiv}

\newboolean{IsArxiv}\setboolean{IsArxiv}{true}

\usepackage{relsize,xspace}
 \usepackage{xcolor}
 \usepackage{mathtools}
 \usepackage{todonotes}
 \usepackage{comment}
\usepackage{microtype}
\usepackage{amsmath}
\usepackage{amsthm}
\usepackage{amssymb}
\usepackage{amsfonts}
\usepackage{stmaryrd}
\usepackage{tikz}
\usepackage{bm}
\usepackage{tikz}
\usepackage{refcount}
\usepackage{wrapfig}
\usepackage{marginnote}
\usepackage{enumerate}
\usepackage{latexsym}
\usepackage{pdfpages}

\ifthenelse{\boolean{IsArxiv}}{
\definecolor{dark-blue}{rgb}{0.05,0.25,0.85}
\hypersetup{colorlinks=true,linkcolor=dark-blue,citecolor=dark-blue}
\usepackage{hyperref}
\usepackage[nameinlink]{cleveref}}{}

\newtheorem*{approxpreprocessing}{Approximate polynomial time pre-processing}
\newtheorem*{approxkernel}{Approximate kernelization}
\newtheorem*{subdiv}{Shallow subdivisions}
\newtheorem*{nd}{Nowhere denseness}
\newtheorem*{dom}{Domination and distance-$r$ domination}
\newtheorem*{domcore}{Domination core}
\newtheorem*{coveringfam}{Covering family}
\newtheorem*{steiner}{Steiner tree}
\newtheorem*{wcoldef}{Weak coloring numbers}
\newtheorem*{avoidingpath}{$A$-avoiding path}
\newtheorem*{projprofile}{Projection profile}
\newtheorem*{graphgp}{The graph $G'$}
\newtheorem*{graphgdot}{The graph $\dot{G}$}

\ifthenelse{\boolean{IsArxiv}}{
\Crefformat{page}{#2page~#1#3}%
\Crefformat{page}{#2Page~#1#3}%
\Crefformat{equation}{#2(#1)#3}%
\Crefformat{equation}{#2(#1)#3}%
\Crefformat{figure}{#2Figure~#1#3}%
\Crefformat{figure}{#2Figure~#1#3}%
\Crefformat{section}{#2Section~#1#3}
\Crefformat{section}{#2Section~#1#3}
\Crefformat{chapter}{#2Chapter~#1#3}
\Crefformat{chapter}{#2Chapter~#1#3}
\Crefformat{chapter*}{#2Chapter~#1#3}
\Crefformat{chapter*}{#2Chapter~#1#3}
\Crefformat{part}{#2Part~#1#3}
\Crefformat{part}{#2Part~#1#3}
\Crefformat{enumi}{#2(#1)#3}
\Crefformat{enumi}{#2(#1)#3}
\Crefformat{algorithm}{#2algorithm~#1#3}
\Crefformat{algorithm}{#2Algorithm~#1#3}
\Crefformat{definition}{#2Definition~#1#3}
\Crefformat{definition}{#2Definition~#1#3}}{}

\newcommand{\wcol}{\mathrm{wcol}}

\newcommand{\WReach}{\mathrm{WReach}}

\newcommand{\Oof}{\mathcal{O}}
\newcommand{\CCC}{\mathcal{C}}

\newcommand{\FFF}{\mathcal{F}}
\newcommand{\GGG}{\mathcal{G}}
\newcommand{\QQQ}{\mathcal{Q}}
\newcommand{\YYY}{\mathcal{Y}}

\newcommand{\fwcol}{f_{\wcol}}

\newcommand{\fproj}{f_{\mathrm{proj}}}

\newcommand{\Opt}{\mathrm{Opt}}

\newcommand{\N}{\mathbb{N}}
\newcommand{\R}{\mathbb{R}}

\renewcommand{\phi}{\varphi}
\renewcommand{\epsilon}{\varepsilon}

\newcommand{\projprof}{\widehat{\mu}}

\title{Lossy kernels for connected distance-$r$ domination on nowhere dense graph classes
\footnote{The work of Sebastian Siebertz is supported by the National Science Centre of Poland via POLONEZ grant agreement UMO-2015/19/P/ST6/03998, 
which has received funding from the European Union's Horizon 2020 research and 
innovation programme (Marie Sk\l odowska-Curie grant agreement No.\ 665778).
}
}
\titlerunning{Lossy kernels for connected distance-$r$ domination on nowhere dense graph classes}
\author{Sebastian Siebertz}
\affil{University of Warsaw, Poland,\\ \texttt{siebertz@mimuw.edu.pl}}
\authorrunning{S.\ Siebertz} 

\Copyright{Sebastian Siebertz}

\begin{document}

\maketitle
\input{abstract}
\input{intro}
\input{kernel}

\input{tree-closure}
\input{lower-bounds}
\input{kernel_conclusion}

\bibliographystyle{plainurl}


\end{document}

%% file: abstract.tex
\begin{abstract}
For $\alpha\colon\N\rightarrow\R$, an $\alpha$-approximate bi-kernel
 is a 
polynomial-time algorithm that takes as input an instance $(I, k)$ of a problem $\QQQ$ and outputs an 
instance $(I',k')$ of a problem $\QQQ'$ of size bounded by a function of $k$ such that, 
for every $c\geq 1$, a $c$-approximate solution for the new instance can be turned 
into a $c\cdot\alpha(k)$-approximate solution of the original instance in polynomial time. This
framework of \emph{lossy kernelization} was recently introduced by Lokshtanov et al.~\cite{lokshtanov2016lossy}. 

We prove that for every nowhere dense class of graphs, 
every $\alpha>1$ and $r\in\N$ there
exists a polynomial $p$ (whose degree depends only on 
$r$ while its coefficients depend on $\alpha$) such that
the connected distance-$r$ dominating set problem 
with parameter $k$ admits an 
$\alpha$-approximate bi-kernel of size $p(k)$.

Furthermore, we show that this result cannot be extended to more general 
classes of graphs which are closed under taking subgraphs 
by showing that if a class $\CCC$ is somewhere dense and 
closed under taking subgraphs, then for some value of $r\in\N$ 
there cannot exist an 
$\alpha$-approximate bi-kernel for the (connected) distance-$r$
dominating set problem on $\CCC$ for any function $\alpha\colon\N\rightarrow\R$ (assuming the Gap Exponential Time Hypothesis). 
\end{abstract}

%% file: intro.tex
\section{Introduction}

\subparagraph*{Lossy kernelization.}

A powerful method in parameterized complexity theory is to compute 
on input $(I,k)$ a problem \emph{kernel} in a polynomial time pre-computation step, 
that is, to reduce the input instance in polynomial time to an equivalent instance $(I',k')$ 
of size $g(k)$ for some function $g$ 
bounded in the parameter only. If the reduced instance $(I',k')$ belongs to a different problem than 
$(I,k)$, we speak of a \emph{bi-kernel}. It is well known that a problem is fixed-parameter tractable if and only if it admits a kernel, however, in general the function~$g$ can grow arbitrarily fast. For practical applications we are mainly interested in linear or at worst polynomial kernels. We refer to the textbooks~\cite{cygan2015parameterized,downey2013fundamentals,
downey1999parameterized} for extensive background on parameterized complexity
and kernelization. 

One shortcoming of the above notion of kernelization is that it does not combine well with approximations or heuristics. An approximate solution on the reduced instance 
provides no insight whatsoever about the original instance, the only 
statement we can derive from the definition of a kernel is that the reduced instance
$(I',k')$ is a positive instance if and only if the original instance $(I, k)$ is a positive instance. This issue was recently addressed by Lokstanov et al.~\cite{lokshtanov2016lossy}, who introduced the framework of \emph{lossy kernelization}. Intuitively, the framework combines notions from approximation and
kernelization algorithms to allow for approximation preserving kernels. 

Formally, a \emph{parameterized optimization (minimization or maximization) problem} $\Pi$ over finite vocabulary $\Sigma$ is a computable function $\Pi\colon\Sigma^\star\times \N \times \Sigma^\star\rightarrow \R\cup\{\pm \infty\}$.  A \emph{solution} for an instance $(I,k)\in \Sigma^\star\times\N$ is a string $s\in \Sigma^\star$, such that $|s| \leq  |I| + k$. The \emph{value} of the solution $s$
is $\Pi(I,k,s)$. For a minimization problem, the \emph{optimum value} of an 
instance $(I,k)$ is $\Opt_\Pi(I,k)=\min_{s\in \Sigma^*, |s|\leq |I|+k}\Pi(I,k,s)$, 
for a maximization problem it is $\Opt_\Pi(I,k)=\max_{s\in \Sigma^*, |s|\leq |I|+k}\Pi(I,k,s)$. An \emph{optimal solution} is a solution $s$ with $\Pi(I,k,s)=\Opt_\Pi(I,k)$. If $\Pi$ is clear from the context, we simply write $\Opt(I,k)$. 

A vertex-subset graph problem~$\mathcal{Q}$ defines which subsets of 
the vertices of an input graph are feasible solutions. We consider the following
parameterized minimization problem associated with $\mathcal{Q}$:
\[\mathcal{Q}(G,k,S)=
\begin{cases}
\infty & \text{if $S$ is not a valid solution for $G$ as determined by $\mathcal{Q}$}\\\min\{|S|,k+1\} & \text{otherwise.}
\end{cases}\]
Note that this bounding of the objective function at $k+1$ does not 
make sense for approximation algorithms if one insists on $k$ being the 
unknown optimum solution of the instance $I$. The parameterization 
above is by the value of the solution that we want our algorithms
to output.

\begin{approxpreprocessing}
Let $\alpha\colon\N\rightarrow \R$ be a function and let $\Pi$ be a parameterized 
minimization problem. An \emph{$\alpha$-approximate polynomial time pre-processing algorithm} $\mathcal{A}$ for $\Pi$ is a pair of polynomial time algorithms. 
The first algorithm is called the \emph{reduction algorithm}, and computes a map
$R_\mathcal{A} \colon \Sigma^\star\times \N\rightarrow \Sigma^\star\times \N$. 
Given as input an instance $(I, k)$ of $\Pi$, the reduction algorithm outputs
another instance $(I',k')=R_\mathcal{A}(I,k)$. 
The second algorithm is called the \emph{solution lifting algorithm}. It takes as
input an instance $(I,k)\in \Sigma^\star\times \N$, the output instance $(I',k')$
of the reduction algorithm, and a solution $s'$ to the instance $(I',k')$. 
The solution lifting algorithm works in time polynomial in $|I|,k, |I'|, k'$ and $s'$, 
and outputs a solution $s$ to $(I, k)$ such that 
\begin{align*}
\frac{\Pi(I,k,s)}{\Opt(I,k)}\leq \alpha(k)\cdot \frac{\Pi(I',k',s')}{\Opt(I',k')}. 
\end{align*}
\end{approxpreprocessing}


\begin{approxkernel}
An \emph{$\alpha$-approximate kernelization algorithm} is an $\alpha$-approximate
polynomial time pre-processing algorithm for which we can
prove an upper bound on the size of the output instances in terms of the parameter of the instance to be pre-processed. We speak of a linear or polynomial 
kernel, if the size bound is linear or polynomial, respectively. If we allow
the reduced instance to be an instance of another problem, we speak of 
an \emph{$\alpha$-approximate bi-kernel}.
\end{approxkernel}

We refer to the work of Lokshtanov et al.~\cite{lokshtanov2016lossy}
for an extensive discussion of related work and examples of problems that
admit lossy kernels. 

\vspace{-2mm}
\subparagraph*{Nowhere denseness and domination.}

The notion of nowhere denseness was 
introduced by Ne\v set\v ril and
Ossona de Mendez~\cite{nevsetvril2010first,nevsetvril2011nowhere} as
a general model of \emph{uniform sparseness} of graphs. Many
familiar classes of sparse graphs, like planar
graphs, graphs of bounded tree-width, graphs of bounded degree, 
and all classes that exclude a fixed (topological) minor, are nowhere
dense. An important and related concept is the notion of a graph class of \emph{bounded expansion}, which was also introduced by 
Ne\v set\v ril and
Ossona de Mendez~\cite{nevsetvril2008grad,nevsetvril2008gradb,nevsetvril2008gradc}. 
Before we give the formal definitions, we remark that all graphs in this paper are finite, undirected and simple. We refer to the textbook~\cite{diestel2012graph} for all undefined notation. 

\begin{subdiv}
Let $H$ be a graph and let $r\in \N$. An \emph{$r$-subdivision} of $H$ is obtained by replacing all edges of $H$
by internally vertex disjoint paths of length (exactly) $r$.  We 
write~$H_r$ for the $r$-subdivision of $H$. 
\end{subdiv}

\begin{nd}
  A class $\CCC$ of graphs is \emph{nowhere dense} if there exists a
  function $t\colon \N\rightarrow \N$ such that for
  all $r\in\N$ and for all $G\in \CCC$ we do not find the
  $r$-subdivision of the complete graph $K_{t(r)}$ as a subgraph of $G$. Otherwise, $\CCC$ is called \emph{somewhere dense}.
\end{nd}

Nowhere denseness turns out to be a very robust concept with several
seemingly unrelated natural characterizations.  These include
characterizations by the density of shallow (topological)
minors~\cite{nevsetvril2010first,nevsetvril2011nowhere},
quasi-wideness~\cite{nevsetvril2011nowhere}, low tree-depth
colorings~\cite{nevsetvril2008grad}, generalized coloring
numbers~\cite{zhu2009coloring}, sparse neighborhood
covers~\cite{GroheKRSS15,grohe2014deciding}, by so-called splitter games~\cite{grohe2014deciding} and by the model-theoretic
concepts of stability and independence~\cite{adler2014interpreting}.
%
%
%
%
For extensive background we refer to the textbook
of Ne\v{s}et\v{r}il and Ossona de Mendez~\cite{sparsity}.

\begin{dom}
In the parameterized \emph{dominating set problem} we are given a
graph~$G$ and an integer parameter $k$, and the objective is
to determine the existence of a subset $D\subseteq V(G)$ of size at
most $k$ such that every vertex $u$ of $G$ is \emph{dominated} by
$D$, that is, either $u$ belongs to~$D$ or has a neighbor in~$D$.
More generally, for fixed $r\in \N$, in the \emph{distance-$r$ 
dominating set problem}
we are asked to determine the existence of a subset~$D\subseteq V(G)$ of size at most
$k$ such that every vertex $u\in V(G)$ is within distance at most~$r$
from a vertex of~$D$. In the \emph{connected (distance-$r$) dominating 
set problem} we additionally demand that the (distance-$r$) dominating
set shall be connected. 
\end{dom}

The dominating set problem plays a central role in the theory of
parameterized complexity, as it is a prime example of a
$\mathsf{W}[2]$-complete problem with the size of the optimal solution as the parameter, thus considered intractable in full generality.
For this reason, the (connected) dominating set problem and 
\mbox{distance-$r$} dominating set problem
have been extensively studied on restricted graph classes. 
A particularly fruitful line of research in this area concerns kernelization
for the (connected) dominating set problem~\cite{alber2004polynomial,bodfomlok+09,fomin10,fomin2012linear,FominLST13,philip2012polynomial}. 
For the more general distance-$r$ dominating set problem 
we know the following results. Dawar and Kreutzer~\cite{DawarK09} showed that for every $r\in \N$ and 
every nowhere dense class $\CCC$,
the distance-$r$ dominating set problem is fixed-parameter
tractable on $\CCC$. 
Drange et al.~\cite{drange2016kernelization} gave a linear bi-kernel for distance-$r$ dominating sets on any graph class of bounded expansion for every $r\in \N$,
and a pseudo-linear kernel for dominating sets on any nowhere dense graph class; that is, a kernel of size $\Oof(k^{1+\epsilon})$, where the $\Oof$-notation hides constants depending on $\epsilon$. 
Precisely, the kernelization
algorithm of Drange et al.~\cite{drange2016kernelization} outputs an instance of an annotated problem where some vertices are not required to be dominated; this will be the case in the present paper as well. Kreutzer et al.~\cite{KreutzerPRS16} provided
a polynomial bi-kernel for the distance-$r$ dominating set problem on every
nowhere dense class for every fixed $r\in \N$ and finally, Eickmeyer et al.~\cite{eickmeyer2016neighborhood} could prove the existence of pseudo-linear bi-kernels of size 
$\Oof(k^{1+\epsilon})$, where the $\Oof$-notation hides constants depending on~$r$ and $\epsilon$. 

It is known that bounded expansion classes of graphs are the limit for the 
existence of polynomial kernels for the connected dominating set problem. 
Drange et al.~\cite{drange2016kernelization} gave an example of a 
subgraph-closed class of bounded expansion which does not admit a 
polynomial kernel for connected dominating sets unless $\mathsf{NP}\subseteq \mathsf{coNP/Poly}$. They also showed that 
nowhere dense classes are the limit for the fixed-parameter tractability 
of the distance-$r$ dominating set problem if we assume closure under
taking subgraphs (in the following, classes which are closed under
taking subgraphs will be called \emph{monotone classes}). 


\subparagraph*{Our results.}

We prove that for every nowhere dense class of graphs, 
every $\alpha>1$ and $r\in\N$ there
exists a polynomial $p$ (whose degree depends only on 
$r$ while its coefficients depend on~$\alpha$) such that
the connected distance-$r$ dominating set problem with
parameter $k$ admits an 
$\alpha$-approximate bi-kernel of size $p(k)$.
Our result extends an earlier result by Eiben et al.~\cite{eiben2017}, who
proved that the connected dominating set problem admits $\alpha$-approximate
bi-kernels of linear size on classes of bounded expansion. Note that 
due to the before mentioned hardness result of connected dominating
set on classes of bounded expansion we cannot expect to obtain 
an $\alpha$-approximate bi-kernel of polynomial size for $\alpha=1$, 
as this lossless bi-kernel would in particular imply the existence of a
polynomial bi-kernel for the problem. However, our proof can easily be adapted to provide 
$\alpha$-approximate bi-kernels for $\alpha=1$ for the
distance-$r$ dominating set problem. 

Our proof follows the approach of 
Eiben et al.~\cite{eiben2017} for connected dominating set $(r=1)$ on 
classes of bounded expansion. 
First, we compute a small set 
$Z\subseteq V(G)$ of 
vertices, called a \emph{$(k,r)$-domination core}, such that every
set of size at most~$k$ which $r$-dominates $Z$ will also be 
a distance-$r$ dominating set of $G$. 
The existence of a $(k,r)$-domination core on nowhere 
dense graph classes of size 
polynomial in $k$ was recently proved by Kreutzer et al.~\cite{siebertz2016polynomial}. We remark that the notion 
of a $c$-exchange domination core for a constant $c$, 
which was used by
Eiben et al.~\cite{eiben2017}, cannot be applied in the
nowhere dense setting, as the constant
$c$ must be chosen in relation
to the edge density of shallow subdivisions, an invariant that
can
be unbounded in nowhere dense classes.

Having found a domination core of size polynomial in $k$,  
the next step is to reduce the number of dominators, i.e.~vertices whose 
role is to dominate other vertices, and the
number of connectors, i.e.~vertices whose role is to connect the solution. We apply the techniques of Eiben et al.~\cite{eiben2017} based on approximation techniques
for the Steiner Tree problem. The main difficulty at this point is to find a 
polynomial bounding the size of the lossy kernel whose degree
is independent of $\alpha$. 

Finally, we prove that this result cannot be extended to more general 
classes of graphs which are monotone
by showing that if a class $\CCC$ is somewhere dense and 
monotone, then for some value of $r\in\N$ 
there cannot exist an 
$\alpha$-approximate bi-kernel for the (connected) distance-$r$
dominating set problem on $\CCC$ for any function $\alpha\colon\N\rightarrow\R$ (assuming the Gap Exponential Time Hypothesis). 
These lower bounds are based on an equivalence between 
FPT-approximation algorithms and approximate kernelization
which is proved in~\cite{lokshtanov2016lossy} and a
result of Chalermsook et al.~\cite{chalermsook17}
stating that FPT-$\alpha(k)$-approximation algorithms 
for the dominating set problem do not exist for any function 
$\alpha$ (assuming the Gap Exponential Time Hypothesis). 

\vspace{-5mm}
\subparagraph*{Organization.}
This paper is organized as follows. In \Cref{sec:kernel} and 
\Cref{sec:tree-closure} we 
prove our positive results. We have split the proof into 
one part which requires no knowledge of nowhere dense 
graph classes and which is proved in \Cref{sec:kernel}. 
In the proof we assume just one lemma which contains
the main technical contribution of the paper and which 
requires more background from nowhere dense graphs. 
The lemma is proved in \Cref{sec:tree-closure}. 
In \Cref{sec:lower-bounds} we prove our lower bounds. 

%% file: kernel.tex
\section{Building the lossy kernel}\label{sec:kernel}

Our notation is standard, we refer to the textbook~\cite{diestel2012graph} for all undefined notation. 
In the following, we fix a nowhere dense class $\CCC$ of 
graphs, $k,r\in \N$ and $\alpha>1$. Furthermore, let
$t= \frac{\alpha-1}{4r+2}$. As we deal with the connected
distance-$r$ dominating set problem we may assume
that all graphs in $\CCC$ are connected. 

\begin{domcore}
Let $G$ be a graph. A set $Z\subseteq V(G)$ is a \emph{$(k,r)$-domination core} for $G$ if every set $D$ of size at most $k$ that $r$-dominates $Z$ also $r$-dominates $G$ 
\end{domcore}

Domination cores of polynomial size exist for nowhere dense
classes, as the following lemma shows. 

\begin{lemma}[Kreutzer et al.~\cite{siebertz2016polynomial}]
\label[lemma]{lem:findcore1}
There exists a polynomial $q$ (of degree depending only on~$r$) and a polynomial time algorithm
  that, given a graph $G\in\CCC$ and $k\in\N$
either correctly concludes that $G$ cannot be $r$-dominated by a 
set of at most $k$ vertices, or finds a $(k,r)$-domination core $Z\subseteq V(G)$ of $G$ of size at most $q(k)$. 
\end{lemma}

We remark that the non-constructive
polynomial bounds that follow from~\cite{siebertz2016polynomial}
can be replaced by much improved constructive bounds~\cite{pilipczuk2017wideness}. 

\medskip
We will work with the following parameterized
minimization variant of the connected distance-$r$ dominating set
problem. 

\[\textsc{CDS}_r(G,k,D)=
\begin{cases}
\infty & \text{if $D$ is not a connected distance-$r$}\\
&\quad \text{dominating set of $G$}\\\min\{|D|,k+1\} & \text{otherwise.}
\end{cases}\]

As indicated earlier, we compute only a bi-kernel and reduce
to the following annotated version of the connected
distance-$r$ dominating set problem. 

\[\textsc{ACDS}_r((G,Z),k,D)=
\begin{cases}
\infty & \text{if $D$ is not a connected distance-$r$}\\
&\quad \text{dominating set of $Z$ in $G$}\\\min\{|D|,k+1\} & \text{otherwise.}
\end{cases}\]

\medskip

The following lemma is 
folklore for dominating sets, its more general variant
for distance-$r$ domination is proved just as the 
case $r=1$ (see e.g.~Proposition 1 of~\cite{eiben2017} 
for a proof for the case $r=1$). 

\begin{lemma}\label[lemma]{lem:ds-cds}
Let $G$ be a graph, $Z\subseteq V(G)$ a connected set in $G$ and 
let $D$ be a distance-$r$ dominating set for $Z$ such that $G[D]$ has at most 
$p$ connected components. Then we can compute in polynomial time 
a set $Q$ of size at most
$2rp$ such that $G[D \cup Q]$ is connected.
\end{lemma}

\medskip
The lemma implies that we may assume that our
domination cores are connected. 

\begin{corollary}\label[corollary]{lem:findcore}
There exists a polynomial $q$ (of degree depending
only on $r$) and a polynomial time algorithm
that, given a graph $G\in\CCC$ and $k\in\N$
either correctly concludes that~$G$ cannot be $r$-dominated by a 
set of at most $k$ vertices, or finds a 
$(k,r)$-domination core $Z\subseteq V(G)$ of~$G$ of size at most $q(k)$ such that $G[Z]$ is connected. 
\end{corollary}
\begin{proof}
Assume that when applying \Cref{lem:findcore1}, a 
$(k,r)$-domination core $Y$ is returned, otherwise we
return that no distance-$r$ dominating set of size at most
$k$ exists. 

First observe that every superset $X\supseteq Y$
is also a $(k,r)$-domination core of $G$ (every set of size at most $k$ which $r$-dominates $X$ in particular $r$-dominates $Y$, and
hence all of $G$). 

Assume there is a vertex $v\in V(G)$ with distance greater 
than $2r$ from $Y$. Since $Y$ is a $(k,r)$-domination core, 
every set of size at most $k$ that $r$-dominates $Y$ also $r$-dominates
$G$. If there exists a distance-$r$ dominator $A$ of $Y$ 
of size at most $k$, also $B=N_r[Y]\cap A$ 
(the intersection of~$A$ with the
closed $r$-neighborhood of $Y$) is a distance-$r$ dominator of 
$Y$ of size at most $k$. However, as $v$ has distance 
greater than $2r$ from $Y$, $B$ cannot be a distance-$r$
dominating set of $G$. Hence, if there is $v\in V(G)$ with
distance greater than $2r$ from $Y$, we may return 
that~$G$ cannot be $r$-dominated by a set of at most $k$
vertices. Otherwise, it follows that $Y$ is a distance-$2r$ 
dominating set of $G$. We can hence apply
\Cref{lem:ds-cds} with parameters $Z=V(G)$ (we assume
that all graphs $G\in\CCC$ are connected) and $D=Y$
to find a connected set $X\supseteq Y$ of size at most
$(2r+1)\cdot q(k)$ which is a connected $(k,r)$-domination core. 
\end{proof}

The key idea is to split connected dominating sets into 
parts of well controlled size. This idea will be realized by
considering covering families, defined as follows. 

\begin{coveringfam}
Let $G$ be a connected graph. A \emph{$(G, t)$-covering family}
is a family $\mathcal{F}(G,t)$ of subtrees of $G$ such that 
for each $T\in \mathcal{F}(G,t)$, $|V(T)|\leq 2t$
and $\bigcup_{T\in \mathcal{F}(G,t)}V(T)=V(G)$. 
\end{coveringfam}

\begin{lemma}[Eiben et al.~\cite{eiben2017}]\label[lemma]{lem:cover}
Let $G$ be a connected graph. There is a
$(G,t)$-covering family $\mathcal{F}(G, t)$ with $|\mathcal{F}(G,t)|\leq 
|V(G)|/t +1$, and $\sum_{T\in \mathcal{F}(G,t)} |V(T)|\leq (1+1/t)\cdot |V(G)|+1$. 
\end{lemma}

To recombine the pieces we will solve instances of the 
\textsc{(Group) Steiner Tree} problem.

\begin{steiner}
Let $G$ be a graph and let $Y\subseteq V(G)$ be a set of \emph{terminals}.
A \emph{Steiner tree} for~$Y$ is a subtree of $G$ spanning $Y$. 
We write $\mathbf{st}_G(Y)$ for the order of (i.e.\ 
the number of vertices of) the smallest Steiner tree for
$Y$ in $G$ (including the vertices of $Y$).

Let $G$ be a graph and let $\YYY=\{V_1,\ldots, V_s\}$ 
be a family of vertex disjoint subsets of $G$. A \emph{group Steiner tree} for $\YYY$ is a subtree of $G$ that contains (at least) one 
vertex of each group~$V_i$. We write $\mathbf{st}_G(\YYY)$ for the order of the smallest group Steiner tree for $\YYY$.
\end{steiner}

\smallskip
When recombining the pieces, we have to preserve their 
domination properties. For this, 
we will need precise a description of how vertices interact 
with the domination core.

\begin{avoidingpath}
Let $G$ be a graph and let $A\subseteq V(G)$ be a subset of vertices. For vertices $v\in A$ and $u\in V(G)\setminus A$, a path $P$ connecting $u$ and $v$ is called {\em{$A$-avoiding}}
if all its vertices apart from $v$ do not belong to $A$. 
\end{avoidingpath}

\begin{projprofile}
The {\em{$r$-projection}} of a vertex $u\in V(G)\setminus A$ onto~$A$, denoted $M^G_r(u,A)$ is the set of all vertices $v\in A$ that
can be connected to $u$ by an $A$-avoiding path of length at most $r$. The {\em{$r$-projection profile}} of a vertex $u\in V(G)\setminus A$ on $A$ is a function $\rho^G_r[u,A]$ mapping vertices of
$A$ to $\{0,1,\ldots,r,\infty\}$, defined as follows: for every $v\in A$, the value $\rho^G_r[u,A](v)$ is the length of a shortest $A$-avoiding path connecting $u$ and~$v$, and~$\infty$ in case this length
is larger than $r$. We define 
\[\projprof_r(G,A)=|\{\rho_r^G[u,A]\colon u\in V(G)\setminus A\}|\]
to be the number of different $r$-projection profiles realized on $A$. 
\end{projprofile}

\begin{lemma}[Eickmeyer et al.~\cite{eickmeyer2016neighborhood}]\label[lemma]{lem:projection-complexity}
There is a function $\fproj$ such that for every
$G\in \CCC$, vertex subset $A\subseteq V(G)$, and
  $\epsilon>0$ 
  we have $\projprof_r(G,A)\leq \fproj(r,\epsilon)\cdot |A|^{1+\epsilon}$.
\end{lemma}

The following lemma is immediate from the definitions. 
\begin{lemma}\label[lemma]{lem:dswithproj}
Let $G$ be a graph and let $X\subseteq V(G)$. Let $D$ be a distance-$r$ dominating set of~$X$. Then every set $D'$ such
that for each $u\in D$ there is $v\in D'$ with 
$\rho_r^G[u,X]=\rho_r^G[v,X]$ is a distance-$r$ dominating 
set of $X$. 
\end{lemma}

The following generalization of the \emph{Tree Closure Lemma}
(Lemma 4.7 of Eiben et al.~\cite{eiben2017}) shows that we 
can  re-combine the pieces in nowhere dense graph classes. 

\begin{lemma}\label[lemma]{lem:tree-closure}
There exists
a function $f$ such that the
following holds. Let $G\in\CCC$, let $X\subseteq V(G)$, 
and let
$\epsilon>0$. Define an equivalence relation 
$\sim_{X,r}$ on $V(G)$ by 
\[u\sim_{X,r}v\Leftrightarrow \rho_r^G[u,X]=\rho_r^G[v,X].\]
Then we can compute in time $\Oof(|X|^{t(1+\epsilon)}\cdot n^{1+\epsilon})$ a subgraph $G'\subseteq G$ of $G$ 
such that 
\begin{align*}
\textit{1)} & \quad X\subseteq V(G'), \\
\textit{2)} & \quad \text{for every $u\in V(G)$ there 
is $v\in V(G')$ with $\rho_r^G[u,X]=\rho_r^{G'}[v,X]$},\\
\textit{3)} & \quad \text{for every set~$\YYY$ of at most $2t$ projection classes (i.e., equivalence classes of $\sim_{X,r}$),}\\
& \quad \quad\text{if $\mathbf{st}_G(\YYY)\leq 2t$, then $\mathbf{st}_{G'}(\YYY)=\mathbf{st}_G(\YYY)$, and }\\
\textit{4)} & \quad |V(G')|\leq f(r,t,\epsilon)\cdot |X|^{2+\epsilon}.
\end{align*}
Note that in item \textit{3)}, due to item \textit{2)}, 
every class of $\sim_{X,r}$ which is non-empty
in $G$ is also a non-empty class of $\sim_{X,r}$ in $G'$. 
\end{lemma}

We defer the proof of the lemma to the next section.

\begin{lemma}\label[lemma]{lemma:pre-kernel}
Let $\epsilon>0$ and let $q$ be the 
polynomial from \Cref{lem:findcore}. There exists an algorithm running in time $\Oof(q(k)^{t(1+\epsilon)}\cdot n^{1+\epsilon})$ that, given an $n$-vertex graph 
  $G\in \CCC$ and a positive integer $k$, either returns that there exists
  no connected distance-$r$ dominating set of $G$, or 
  returns a subgraph $G'\subseteq G$ and a vertex subset $Z\subseteq V(G')$ with the following properties:
  \begin{align*}
  \textit{1)} &\quad \text{$Z$ is a $(k,r)$-domination 
  core of $G$,}\\
  \textit{2)}&\quad \text{$\Opt_{\textsc{ACDS}_r}((G',Z),k)\leq 
  \alpha\cdot \Opt_{\textsc{CDS}_r}(G,k)$, and}\\
  \textit{3)}&\quad \text{$|V(G')|\leq p(k)$, for some polynomial $p$ whose degree depends only on $r$.}
  \end{align*}
\end{lemma}
\begin{proof}
Using \Cref{lem:findcore}, we first conclude that $G$ cannot be $r$-dominated by a 
connected set of at most $k$ vertices, or we find a 
connected $(k,r)$-domination core $Z\subseteq V(G)$ of $G$ of size at most 
$q(k)$.
In the first case, we reject the instance, otherwise, 
let $G'\subseteq G$ be the subgraph that we obtain 
by applying \Cref{lem:tree-closure} with parameters
$G,Z,t$ and $\epsilon$. Let $p\coloneqq
f(r,t,\epsilon)\cdot q^{2+\epsilon}$ (where~$f$
is the function from \Cref{lem:tree-closure}), which is a polynomial
of degree depending only on~$r$, 
only the coefficients depend on $\alpha$. 

It remains to show that $\Opt_{\textsc{ACDS}_r}((G',Z),k)\leq \alpha\cdot \Opt_{\textsc{CDS}_r}(G,k)$. Let $D^*$ be a minimum connected distance-$r$ dominating set of $G$ of size at most $k$ (if $|D^*|>k$, then $\Opt_{\textsc{ACDS}_r}((G',Z),k)\leq \alpha\cdot \Opt_{\textsc{CDS}_r}(G,k)$ trivially holds). Let $\mathcal{F}=\mathcal{F}(G[D^*],t)=\{T_1,\ldots, T_\ell\}$ 
be a covering family for the connected graph $G[D^*]$ 
obtained by \Cref{lem:cover}. Note that by the lemma we 
have $\ell\leq |D^*|/t+1$ and $\sum_{1\leq i\leq \ell}V(T_i)\leq (1+1/t)|D^*|+1$. 
Moreover, the size of each subtree $T_i$ is at most $2t$. 
By construction of $G'$ (according to item \textit{3)} of 
\Cref{lem:tree-closure}), 
for each $T\in \mathcal{F}$ there exists a tree $T'$ in $G'$ 
of size at most $|V(T)|$ which contains for each $u\in V(T)$
a vertex $v$ with $\rho_r^G[u,Z]=\rho_r^{G'}[v,Z]$.

We construct a new family $\mathcal{F}'$ which we obtain by replacing each $T\in \mathcal{F}$ by the tree $T'$ described
above. Let $D'\coloneqq \bigcup_{T'\in \mathcal{F}'}V(T')$ in $G'$. 
We have \mbox{$\sum_{T'\in \mathcal{F}'}|V(T')|\leq (1+1/t)|D^*|+1$} and
since $D'$ contains
vertices from the same projection classes as $D^*$, according to \Cref{lem:dswithproj},~$D'$ is a distance-$r$ dominating set of $Z$. Moreover, $G[D']$ has at 
most $\ell\leq |D^*|/t+1$ components. We apply \Cref{lem:ds-cds}, and obtain
a set $Q$ of size at most $2r(|D^*|/t+1)$ such that 
$D''=D'\cup Q$ is
a connected distance-$r$ dominating set of $Z$. 
We hence have \[|D''|\leq 2r(|D^*|/t+1)+(1+1/t)|D^*|+1=(1+\frac{2r+1}{t})|D^*|+2r+1\leq (1+\frac{4r+2}{t})|D^*|\] (we may assume
that $2r+1\leq \frac{2r+1}{t}|D^*|$, as otherwise we can simply run a brute force algorithm in polynomial time). We conclude by 
recalling that $t= \frac{\alpha-1}{4r+2}$.
\end{proof}

\begin{theorem}\label{thm:lossykernel}
There exists a polynomial $p$ whose
degree depends only on $r$ such that 
the connected distance-$r$ dominating set problem on $\CCC$
admits an $\alpha$-approximate bi-kernel with 
$p(k)$ vertices. 
\end{theorem}
\begin{proof}
The $\alpha$-approximate polynomial time pre-processing
algorithm first calls the algorithm of
\Cref{lemma:pre-kernel}. If 
it returns that there exists no distance-$r$
dominating set of size at most $k$ for $G$, we return 
a trivial negative instance. Otherwise, 
let $((G',Z),k)$ be the annotated instance returned by 
the algorithm. 
The solution lifting algorithm, 
given a connected distance-$r$ dominating set of $Z$
in $G'$ simply returns $D$.

By construction
of $G'$ we have $M_r^{G'}(u,Z)\subseteq M_r^G(u,Z)$ 
for all $u\in V(G')$. Hence every distance-$r$ $Z$-dominator
in $G'$ is also a distance-$r$ $Z$-dominator in $G$. 
In particular, since $Z$ is a 
$(k,r)$-domination core, $D$ is also a connected 
distance-$r$ dominating set for $G$. 

Finally, by \Cref{lemma:pre-kernel}
we have $\Opt_{\textsc{ACDS}_r}((G',Z),k)\leq 
\alpha\cdot \Opt_{\textsc{CDS}_r}(G,k)$, which implies
\begin{align*}
\frac{{\textsc{CDS}_r}(G,k,D)}{\Opt_{\textsc{CDS}_r}(G,k)}\leq \alpha(k)\cdot \frac{{\textsc{ACDS}_r}((G',Z),k,D)}{\Opt_{\textsc{ACDS}_r}((G',Z),k)}. \hspace{6.1cm}\qedhere
\end{align*}
\end{proof}

Observe that we obtain a $1$-approximate bi-kernel for
the distance-$r$ dominating set problem by just taking one 
vertex from each projection class of the $(k,r)$-domination
core.

%% file: tree-closure.tex
\section{The proof of \Cref{lem:tree-closure}}\label{sec:tree-closure}

\Cref{lem:tree-closure} is the most technical contribution
of this paper. This whole section is devoted to its proof. 
We will mainly make use of a characterization of nowhere 
dense graph classes by the so-called \emph{weak 
coloring numbers}. 

\begin{wcoldef}
For a graph $G$, by $\Pi(G)$ we denote the set of all linear orders
of $V(G)$.  For $u,v\in V(G)$ and
any $s\in\N$, we say that~$u$ is \emph{weakly $s$-reachable} from~$v$
with respect to~$L$, if there is a path $P$ of length at most $s$ connecting $u$ and $v$ such that $u$ is 
the smallest among
the vertices of $P$ with respect to~$L$. By $\WReach_s[G,L,v]$ we
denote the set of vertices that are weakly $s$-reachable from~$v$ with
respect to~$L$. For any subset $A\subseteq V(G)$, we let
$\WReach_s[G,L,A] = \bigcup_{v\in A} \WReach_s[G,L,v]$.  The
\emph{weak $s$-coloring number $\wcol_s(G)$} of $G$ is defined as 
\begin{eqnarray*}
  \wcol_s(G)& = & \min_{L\in\Pi(G)}\:\max_{v\in V(G)}\:
                   \bigl|\WReach_s[G,L,v]\bigr|.
\end{eqnarray*}
\end{wcoldef}

The weak coloring numbers were introduced by Kierstead and
Yang~\cite{kierstead2003orders} in the context of coloring and
marking games on graphs. As proved by Zhu \cite{zhu2009coloring},
they can be used to characterize both bounded expansion and nowhere
dense classes of graphs. In particular, we use the following.

\begin{theorem}[Zhu \cite{zhu2009coloring}]\label{lem:wcolbound}
  Let $\CCC$ be a nowhere dense class of graphs.
  There is a function $f_{\wcol}$ such that
  for
  all $s\in\N$, $\epsilon>0$, and $H\subseteq G\in \CCC$ we have
  $\wcol_s(H)\leq f_{\wcol}(s,\epsilon) \cdot |V(H)|^\epsilon$.
\end{theorem}

One can define artificial classes where the functions $f_\wcol$ grow 
arbitrarily fast, however, on many familiar sparse graph classes they are
quite tame, e.g.\ on bounded tree-width graphs~\cite{GroheKRSS15}, 
graphs with excluded minors~\cite{siebertz16} or excluded topological
minors~\cite{KreutzerPRS16}. Observe that in any case
the theorem allows to pull polynomial blow-ups on the graph 
size to the function $\fwcol$. More precisely, for any $\epsilon>0$, 
if we deal with a subgraph of size $n^x$ for some $x\in \N$, 
by re-scaling $\epsilon$ to $\epsilon'=\epsilon/x$, we will 
get a bound of $\fwcol(s,\epsilon')\cdot (n^x)^{\epsilon'}
=\fwcol(s,\epsilon')\cdot n^\epsilon$ for the weak $s$-coloring 
number. 

Our second application of the weak coloring numbers is described in the next lemma, which shows that 
they capture the local separation properties of a 
graph. 

\begin{lemma}[see Reidl et al.~\cite{reidl2016characterising}]\label[lemma]{lem:wcol-sep}
Let $G$ be a graph and let $L\in \Pi(G)$. Let $X\subseteq V(G)$, $y\in V(G)$
and let $P$ be a path of length at most $r$ between a vertex $x\in X$ and $y$. 
Then \[\big(\WReach_r[G,L,X]\cap \WReach_r[G,L,y]\big)\cap V(P)\neq \emptyset.\]
\end{lemma}
\begin{proof}
Let $z$ be the minimal vertex of $P$ with respect to $L$. Then both $z\in \WReach_r[G,L,x]$ and $z\in \WReach_r[G,L,y]$. 
\end{proof}

We are now ready to define the graph $G'$ whose
existence we claimed in the previous section. 

\begin{graphgp}\label{def:GX}
Let $G\in\CCC$ and fix a subset $X\subseteq V(G)$. 
Define an equivalence relation 
$\sim_{X,r}$ on $V(G)$ by 
\[u\sim_{X,r}v\Leftrightarrow \rho_r^G[u,X]=\rho_r^G[v,X].\] 
For each subset~$\YYY$ of projection classes 
of size at most $2t$, if $\mathbf{st}_G(\YYY)\leq 2t$, 
fix a Steiner tree~$T_\YYY$ for~$\YYY$ of minimum size.
For such a tree $T_\YYY$ call a vertex $u\in \kappa\cap V(T_\YYY)$ with $\kappa\in \YYY$ a \emph{terminal} of~$T_\YYY$. 
We let $C=\{u\in V(G) : u$ is a terminal of some
$T_\YYY\}$. 

Let~$G'$ be a subgraph of $G$ which contains 
$X$, all $T_\YYY$ as above, and a set of vertices and 
edges such that $\rho_r^G[u,X]=\rho_r^{G'}[u,X]$
for all $u\in C$. 
\end{graphgp}

\begin{lemma}\label[lemma]{lem:computeG_X}
There exist functions $f$ and 
$g$ such that 
for every $G\in\CCC$, $X\subseteq V(G)$ and $\epsilon>0$ 
we can compute 
a graph $G'$ as described above of size $f(r,\epsilon)\cdot
|X|^{2t(1+\epsilon)}$ in time $g(r,t,\epsilon)\cdot
|X|^{2t(1+\epsilon)}$. 
\end{lemma}
\begin{proof}
According to \Cref{lem:projection-complexity} there is a function
$\fproj$ such that for every $G\in \CCC$, vertex subset 
$A\subseteq V(G)$, and $\epsilon>0$ we have $\projprof_r(G,A)\leq \fproj(r,\epsilon)\cdot |A|^{1+\epsilon}$. We now apply the lemma
to $A=X$. 

We compute for each $v\in X$ 
the first $r$ levels of a breadth-first search 
(which terminates whenever 
another vertex of $X$ is encountered, as to compute $X$-avoiding
paths). For each visited vertex $w\in V(G)$ we remember the
distance to $v$. In this manner, we compute in time
$\Oof(|X|\cdot n^{1+\epsilon})$ the projection profile
of every vertex $w\in V(G)$. Observe that \Cref{lem:wcolbound}
applied to $r=1$ implies that an 
$n$-vertex graph $G\in \CCC$ is $n^\epsilon$-degenerate 
and in particular has only $\Oof(n^{1+\epsilon})$ many edges. 
Hence a breadth-first search can be computed in time
$\Oof(n^{1+\epsilon})$. 

We now decide for each subset $\YYY$ of at most $2t$ 
projection classes whether $\mathbf{st}_G(\YYY)\leq 2t$. 
If this is the case, we also compute a 
Steiner tree $T_\YYY$ of minimum size in time 
$h(t,\epsilon)\cdot n^{1+\epsilon}$ for some
function $h$. To see that this
is possible, observe that the problem is equivalent to testing 
whether an existential 
first-order sentence holds in a colored graph, which is possible
in the desired time on nowhere dense classes~\cite{grohe2014deciding, sparsity}.

Finally, for each sub-tree $T_\YYY$ and each $\kappa\in 
\YYY$ fix some terminal $u\in \kappa\cap V(T_\YYY)$. Compute the 
first~$r$ levels of an $X$-avoiding breadth-first search with
root $u$ and add the vertices and edges of the bfs-tree 
to ensure
that $\rho_r^G[u,X]=\rho_r^{G'}[u,X]$. Observe that
by adding these vertices 
we add at most $|X|\cdot r$ vertices for each vertex $u$. 

As we have  $\Oof\left(\left(|X|^{(1+\epsilon)}\right)^{2t}\right)=\Oof\left(|X|^{2t(1+\epsilon)}\right)$ many subsets
of projection classes of size at most $2t$, we can conclude by defining $f$ and $g$ appropriately. 
\end{proof}

It remains to argue that the graph $G'$ is in fact much smaller than 
our initial estimation in \Cref{lem:computeG_X}. First, as 
outlined earlier, we do not care about polynomial blow-ups when 
bounding the weak coloring numbers. 

\begin{lemma}\label[lemma]{lem:wcolGX}
There is a function $h$ such that
for all $s\in \N$ and $\epsilon>0$
we have \[\wcol_{s}(G')\leq h(r,s,t,\epsilon)\cdot |X|^\epsilon.\]
\end{lemma}
\begin{proof}
Choose $\epsilon'\coloneqq \epsilon/(3t)$. According to  \Cref{lem:computeG_X}, $G'$ has size at most 
$f(r,1/2)\cdot |X|^{3t}$ (apply the lemma with $\epsilon=1/2$). 
According to \Cref{lem:wcolbound}, we have 
\[\wcol_{s}(G')\leq \fwcol(s,\epsilon')\cdot \left(f(r,1/2)\cdot |X|^{3t}\right)^{\epsilon'}.\] Conclude by defining $h(r,s,t,\epsilon)=
\fwcol(s,\epsilon')\cdot f(r,1/2)^{\epsilon'}$. 
\end{proof}

Our next aim is to decompose the group Steiner trees
into single paths which are then analyzed with the help of
the weak coloring numbers. We need a few more
auxiliary lemmas. 

\begin{definition}
The \emph{lexicographic product} $G\bullet H$ of two graphs $G$ and
$H$ is defined by $V(G\bullet H)= V(G)\times V(H)$ and $E(G\bullet
H)=\big\{\{(x,y), (x',y')\} : \{x,x'\}\in E(G)$ or $\big(x=x'$ and
$\{y,y'\}\in E(H)\big)\big\}$. 
\end{definition}

The following two lemmas are easy consequences of 
the definitions. 

\begin{lemma}\label[lemma]{lem:wcollex}
Let $G,H$ be graphs and let $s\in \N$. Then 
$\wcol_{s}(G\bullet H)\leq |V(H)|\cdot \wcol_{s}(G)$. 
\end{lemma}


\begin{lemma}\label[lemma]{lem:wcolsubdiv}
Let $G$ be a graph and let $r',s\in\N$. 
Let $H$ be any graph obtained by replacing some
edges of $G$ by paths of length $r$. 
Then $\wcol_{r'}(H)\leq s+\wcol_{r'}(H)$. 
\end{lemma}

To estimate the size of $G'$ we reduce the 
group Steiner tree problems to simple Steiner tree problems
in a super-graph $\dot{G}$ of $G'$. 

\begin{graphgdot}
See Figure 1 for an illustration of the following 
construction. Let $G'$ with distinguished terminal vertices $C$ be as 
described
above. 
For each equivalence class $\kappa$ represented in $C$, 
fix some vertex $x_\kappa\in M_r^G(u,X)$
for $u\in \kappa$ which is of minimum distance to $u$ 
among all such choices (for our purpose we may assume
that the empty class with 
$M_r^G(u,X)=\emptyset$ is not realized in $G$). 

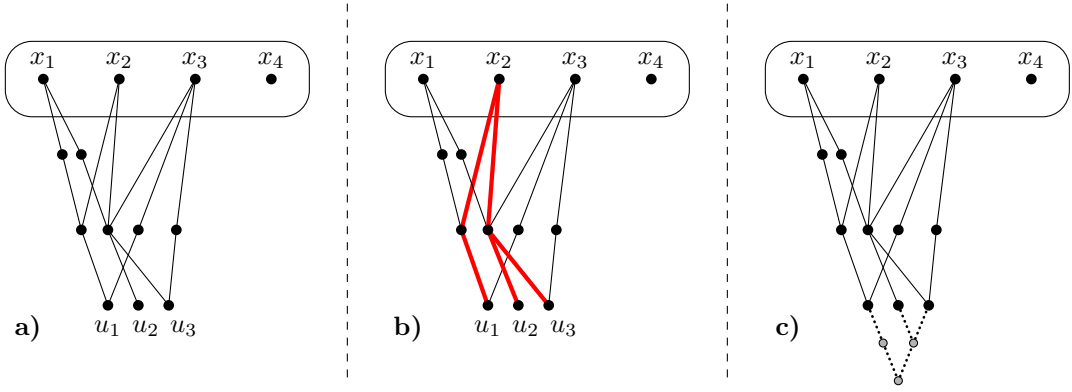
\begin{figure}
\begin{center}
  \begin{tikzpicture}[circle dotted/.style={dash pattern=on .05mm 
  off 2pt, line cap=round}]
  
  \node at (-0.2, -3.3) {\textbf{a)}};
  \fill[black] (0,0) circle (2pt); 
\fill[black] (1,0) circle (2pt); 
\fill[black] (2,0) circle (2pt); 
\fill[black] (3,0) circle (2pt); 
\draw[rounded corners=10] (-0.5,-0.5) rectangle (3.5,0.5);  
\node at (0,0.25) {$x_1$};
\node at (1,0.25) {$x_2$};
\node at (2,0.25) {$x_3$};
\node at (3,0.25) {$x_4$};

\fill[black] (0.25,-1) circle (2pt); 
\fill[black] (0.5,-2) circle (2pt); 
\fill[black] (0.85,-3) circle (2pt); 
\fill[black] (1.25,-2) circle (2pt); 
\draw[-] (0.85,-3) -- (0.5,-2) -- (0.25, -1) -- (0,0); 
\draw[-] (0.5,-2) -- (1,0);
\draw[-] (1.25,-2) -- (2,0);
\draw[-] (0.85,-3) -- (1.25,-2);

\draw[-] (1.25,-3) -- (0.85,-2) -- (0.5, -1) -- (0,0); 
\draw[-] (0.85,-2) -- (1,0);
\draw[-] (0.85,-2) -- (2,0);
\draw[-] (1.75,-2) -- (2,0);
\draw[-] (1.65,-3) -- (0.85,-2);
\draw[-] (1.65,-3) -- (1.75,-2);
\fill[black] (0.5,-1) circle (2pt); 
\fill[black] (0.85,-2) circle (2pt); 
\fill[black] (1.25,-3) circle (2pt); 
\fill[black] (1.75,-2) circle (2pt); 
\fill[black] (1.65,-3) circle (2pt); 

\node at (0.85,-3.3) {$u_1$};
\node at (1.35,-3.3) {$u_2$};
\node at (1.85,-3.3) {$u_3$};

\draw[dashed] (4,1) -- (4,-4);

\begin{scope}[xshift=5cm]
  \node at (-0.2, -3.3) {\textbf{b)}};

\draw[rounded corners=10] (-0.5,-0.5) rectangle (3.5,0.5);  
\node at (0,0.25) {$x_1$};
\node at (1,0.25) {$x_2$};
\node at (2,0.25) {$x_3$};
\node at (3,0.25) {$x_4$};

\draw[-,ultra thick,red] (0.85,-3) -- (0.5,-2);
\draw[-] (0.5,-2) -- (0.25, -1) -- (0,0); 
\draw[-,ultra thick,red] (0.5,-2) -- (1,0);
\draw[-] (1.25,-2) -- (2,0);
\draw[-] (0.85,-3) -- (1.25,-2);

\draw[-,ultra thick,red] (1.25,-3) -- (0.85,-2);
\draw[-] (0.85,-2) -- (0.5, -1) -- (0,0); 
\draw[-,ultra thick,red] (0.85,-2) -- (1,0);
\draw[-] (0.85,-2) -- (2,0);
\draw[-] (1.75,-2) -- (2,0);
\draw[-,ultra thick,red] (1.65,-3) -- (0.85,-2);
\draw[-] (1.65,-3) -- (1.75,-2);

  \fill[black] (0,0) circle (2pt); 
\fill[black] (1,0) circle (2pt); 
\fill[black] (2,0) circle (2pt); 
\fill[black] (3,0) circle (2pt); 
\fill[black] (0.25,-1) circle (2pt); 
\fill[black] (0.5,-2) circle (2pt); 
\fill[black] (0.85,-3) circle (2pt); 
\fill[black] (1.25,-2) circle (2pt); 

\fill[black] (0.5,-1) circle (2pt); 
\fill[black] (0.85,-2) circle (2pt); 
\fill[black] (1.25,-3) circle (2pt); 
\fill[black] (1.75,-2) circle (2pt); 
\fill[black] (1.65,-3) circle (2pt); 

\node at (0.85,-3.3) {$u_1$};
\node at (1.35,-3.3) {$u_2$};
\node at (1.85,-3.3) {$u_3$};
\end{scope}

\draw[dashed] (9,1) -- (9,-4);

\begin{scope}[xshift=10cm]
  \node at (-0.2, -3.3) {\textbf{c)}};

\draw[rounded corners=10] (-0.5,-0.5) rectangle (3.5,0.5);  
\node at (0,0.25) {$x_1$};
\node at (1,0.25) {$x_2$};
\node at (2,0.25) {$x_3$};
\node at (3,0.25) {$x_4$};

\draw[-] (0.85,-3) -- (0.5,-2);
\draw[-] (0.5,-2) -- (0.25, -1) -- (0,0); 
\draw[-] (0.5,-2) -- (1,0);
\draw[-] (1.25,-2) -- (2,0);
\draw[-] (0.85,-3) -- (1.25,-2);

\draw[-] (1.25,-3) -- (0.85,-2);
\draw[-] (0.85,-2) -- (0.5, -1) -- (0,0); 
\draw[-] (0.85,-2) -- (1,0);
\draw[-] (0.85,-2) -- (2,0);
\draw[-] (1.75,-2) -- (2,0);
\draw[-] (1.65,-3) -- (0.85,-2);
\draw[-] (1.65,-3) -- (1.75,-2);

  \fill[black] (0,0) circle (2pt); 
\fill[black] (1,0) circle (2pt); 
\fill[black] (2,0) circle (2pt); 
\fill[black] (3,0) circle (2pt); 
\fill[black] (0.25,-1) circle (2pt); 
\fill[black] (0.5,-2) circle (2pt); 
\fill[black] (0.85,-3) circle (2pt); 
\fill[black] (1.25,-2) circle (2pt); 

\fill[black] (0.5,-1) circle (2pt); 
\fill[black] (0.85,-2) circle (2pt); 
\fill[black] (1.25,-3) circle (2pt); 
\fill[black] (1.75,-2) circle (2pt); 
\fill[black] (1.65,-3) circle (2pt); 


\draw[line width = 1pt,circle dotted] (1.05,-3.5) -- (1.25,-4) -- (1.45,-3.5);
\draw[line width = 1pt,circle dotted] (1.05, -3.5) -- (0.85,-3);
\draw[line width = 1pt,circle dotted] (1.25, -3) -- (1.45,-3.5) -- (1.65,-3);
\fill[black!30!white, draw=black] (1.05,-3.5) circle (1.5pt); 
\fill[black!30!white, draw=black] (1.45,-3.5) circle (1.5pt); 
\fill[black!30!white, draw=black] (1.25,-4) circle (1.5pt); 
\end{scope}

  \end{tikzpicture}
\end{center}
\label[figure]{fig:dotG}
\caption{a) The vertices $u_1,u_2,u_3$ realize the same projection
profile $\rho_r^G[u_1,X]=(3,2,2,\infty)$. \\b) We have chosen 
$x_2$ as $x_\kappa$, which results in the indicated tree $T_\kappa$. c) A subdivided copy of $T_\kappa$ is added to $\dot{G}$.}
\end{figure}

Let $T_\kappa$ be a tree which contains
for each $u\in \kappa\cap C$ an $X$-avoiding path of minimum 
length between $u$ and $x_\kappa$ (e.g.\ obtained by  
an $X$-avoiding breadth-first search with root $x_\kappa$). 
Note that the vertices
of $\kappa\cap C$ appear as leaves of $T_\kappa$ and all 
leaves have the same distance from the root $x_\kappa$.
To see this, note that if a vertex $u$ of
$\kappa\cap C$ lies on a shortest path from~$x_\kappa$ to another
vertex $v$ of~$\kappa\cap C$, then the $X$-avoiding 
distance between $u$ and $x_\kappa$
is smaller than the $X$-avoiding distance between 
$v$ and $x_\kappa$, contradicting that all vertices of 
$\kappa\cap C$ have the same projection profile. Recall that 
by construction projection profiles are preserved for each 
vertex of $\kappa\cap C$.

Let $\dot{G}$ be the graph obtained by adding to $G'$ 
for each equivalence class $\kappa\cap C$ a 
copy of~$T_\kappa$, with each
each edge subdivided~$2r$
times. Then identify the leaves of this copy $T_\kappa$
with the respective vertices of $\kappa$.  
\end{graphgdot}

\begin{lemma}\label[lemma]{lem:classgraph}
There exists a function $f_\bullet$ such that for all 
$r'\in\N$ and all $\epsilon>0$ 
we have $\wcol_{r'}(\dot{G})\leq f_\bullet(r',t, \epsilon)
\cdot |X|^{1+\epsilon}$. 
\end{lemma}
\begin{proof}
Let $\epsilon'\coloneqq \epsilon/2$. 
According to \Cref{lem:projection-complexity}, there is a
function $\fproj$ such that there are at most 
$\fproj(r,\epsilon')\cdot|X|^{1+\epsilon'}\eqqcolon x$ distinct 
projection profiles. When constructing the graph~$\dot{G}$, we hence create at most so many
trees $T_\kappa$. These can be found as disjoint 
subgraphs in $G'\bullet K_x$. Hence, $\dot{G}$ is a 
subgraph of a $2r$-subdivision of
$G'\bullet K_x$. 
According to \Cref{lem:wcolGX}, \Cref{lem:wcollex} 
and \Cref{lem:wcolsubdiv}
we have $\wcol_{r'}(\dot{G})\leq 
h(r,r',t,\epsilon')\cdot |X|^{\epsilon'}\cdot \fproj(r,\epsilon')\cdot 
|X|^{1+\epsilon'}+r'$, where $h$ is the function from 
\Cref{lem:wcolGX}. Assuming that each of these terms
is at least $1$, we can define $f_\bullet(r',t,\epsilon)\coloneqq 
r'\cdot h(r',t,\epsilon')\cdot \fproj(r,\epsilon')$.
\end{proof}

\begin{lemma}\label{lem:translateSteiner}
With each group Steiner tree problem for $\YYY$, we associate
the Steiner tree problem for the set~$Y$ which contains exactly the 
roots of the subdivided trees $T_\kappa$ for each 
$\kappa\in \YYY$. Denote this root by $v_\kappa$
(it is a copy of $x_\kappa$). 
Denote by $d_\kappa$ the distance from $v_\kappa$ 
to $x_\kappa$. Then every group Steiner tree $T_\YYY$ for 
$\YYY$
of size $s\leq 2t$ in $G$ gives rise to a Steiner 
tree for $Y$ of size $s+\sum_{\kappa\in \YYY} d_\kappa$ in 
$\dot{G}$. Vice versa, every Steiner tree for a set~$Y$ of
the above form
of size $s+\sum_{\kappa\in \YYY} d_\kappa$ in $\dot{G}$ 
gives rise to a group Steiner tree of size $s$ for $\YYY$
in $G$. 
\end{lemma}
\begin{proof}
The forward direction is clear. Conversely, 
let $T_Y$ be a Steiner tree for a set $Y$ which contains
only roots of subdivided trees $T_\kappa$
of size 
$s+\sum_\kappa d_\kappa$ in $\dot{G}$. 
We claim that $T_Y$ uses exactly $d_\kappa$
vertices of $T_\kappa$, more precisely, $T_Y$ 
connects exactly one vertex $u\in \kappa$ 
with~$v_\kappa$. Assume $T_Y$ contains two paths $P_1,P_2$
between $v_\kappa$ and vertices $u_1,u_2$ from $\kappa$. 
Because we work with a $2r$-subdivision of $T_\kappa$, 
we have $|V(P_1)\cup V(P_2)|\geq d\kappa+2r$. However, 
there is a path between $u_1$ and $u_2$ via $x_\kappa$
of length at most $2r$ (which uses only $2r-1$ vertices) in 
$\dot{G}$, contradicting the fact that $T_Y$ uses
a minimum number of vertices. 
\end{proof}

\begin{lemma}
There is a function $f$ such that for every $\epsilon>0$
the graph $\dot{G}$ contains at most $f(r,t,\epsilon)\cdot 
|X|^{2+\epsilon}$ vertices. 
\end{lemma}
\begin{proof}
Let $\epsilon'\coloneqq \epsilon/2$. 
Every Steiner tree $T_Y$ that connects a 
subset $Y$ decomposes into paths~$P_{uv}$ between pairs 
$u,v\in Y$. 
According to \Cref{lem:wcol-sep}, each such path $P_{uv}$ 
contains a vertex $z$ which is weakly 
$(4r^2+2t)$-reachable from $u$ and from $v$. 
This is because each Steiner tree in $\dot{G}$ connecting
$u$ and $v$ contains a path of length at most $2r^2$ 
between $u$ and some leaf $u_\kappa\in \kappa\cap C$
(and analogously a path of length at most $2r^2$ 
between $v$ and some leaf $v_\kappa\in \kappa\cap C$). 
Now $u_\kappa$ and $v_\kappa$ are connected by a path
of length at most $2t$ by construction. 

Denote by $Q_u$
and $Q_v$, respectively, 
the sub-path of $P_{uv}$ between $u$ and $z$, and $v$ and $z$, 
respectively. We charge the vertices of $Q_u$ to vertex $u$
and the vertices of $Q_v$ to vertex $v$ (and the vertex $z$ 
to one of the two). According to \Cref{lem:classgraph}, 
each vertex weakly $(4r^2+2t)$-reaches at most $f_\bullet(4r^2+2t,t,\epsilon')
\cdot 
|X|^{1+\epsilon'}$ vertices which can play the role of $z$. 
According to \Cref{lem:projection-complexity} we have 
at most $\fproj(r,\epsilon')\cdot |X|^{1+\epsilon'}$ choices for
$u,v\in Y$. Hence we obtain that all Steiner trees add up to at most 
$\fproj(r,\epsilon')\cdot |X|^{1+\epsilon'}\cdot f_\bullet(4r^2+2t,t,\epsilon')
\cdot 
|X|^{1+\epsilon'}\eqqcolon f(r,t,\epsilon)\cdot |X|^{2+\epsilon}$
vertices. 
\end{proof}

As $G'$ is a subgraph of $\dot{G}$, we conclude
that also $G'$ is small. 

\begin{corollary}
There is a function $f$ such that for every $\epsilon>0$
the graph $\dot{G}$ has size at most $f(r,t,\epsilon)\cdot 
|X|^{2+\epsilon}$. 
\end{corollary}

This was the last missing statement of~\Cref{lem:tree-closure}, 
which finishes the proof. 

%% file: lower-bounds.tex
\section{Lower bounds}\label{sec:lower-bounds}

Our lower bound is based on Proposition~3.2 of \cite{lokshtanov2016lossy} which establishes
 equivalence between FPT-approximation
algorithms and approximate kernelization.

\begin{lemma}[Proposition~3.2 of \cite{lokshtanov2016lossy}]\label[lemma]{lemma:fpt-approx}
For every function $\alpha$ and decidable parameterized 
optimization problem $\Pi$,
$\Pi$ admits a fixed parameter tractable $\alpha$-approximation algorithm if and only if $\Pi$ has an $\alpha$-approximate kernel.
\end{lemma}

We will use 
a reduction from set cover to the distance-$r$ dominating
set problem. Recall that the instance of the Set Cover problem
consists of $(U, \FFF, k)$, where $U$ is a finite universe, $\FFF\subseteq 2^U$ is a family of subsets of the universe, and
$k$ is a positive integer. The question is whether there exists a subfamily $\GGG \subseteq \FFF$ of size $k$ such that every
element of $U$ is covered by~$\GGG$, i.e., $\bigcup G=U$. 
The following result states that under complexity 
theoretic assumptions for the set cover problem
on general graphs
there does not exist a fixed-parameter tractable $\alpha$-approximation algorithm for any function $\alpha$. 

\begin{lemma}[Chalermsook et al.~\cite{chalermsook17}]\label[lemma]{lemma:fpt-approx-lowerbound}
If the Gap Exponential Time Hypothesis (gap-ETH) holds, 
then there is no fixed parameter tractable $\alpha$-approximation algorithm for the
set cover problem, for any function~$\alpha$.
\end{lemma}

By definition of nowhere dense graph classes, if 
$\CCC$ is somewhere dense (that is, not nowhere dense), 
then for some $r\in \N$ we find the $r$-subdivision 
of every graph as a subgraph of a graph in $\CCC$. 
For $p\geq 0$, let $\mathcal{H}_p$ be the class of 
$p$-subdivisions of all simple graphs, that is, the class 
comprising all the graphs that can be obtained from 
any simple graph by replacing every edge by a path of
length $p$. As our definition of nowhere denseness
in the introduction is not the standard definition 
but tailored to the following hardness reduction, 
we give reference to the following lemma. 


\begin{lemma}[Ne\v{s}et\v{r}il and Ossona de Mendez~\cite{nevsetvril2011nowhere}]\label[lemma]{lemma:somewheredense}
For every monotone somewhere dense graph class~$\CCC$, there exists $r\in\N$ such 
that $\mathcal{H}_r\subseteq \CCC$. 
\end{lemma}

Based on the above lemma, 
in the arxiv-version of \cite{drange2016kernelization}, 
a parameterized reduction from set cover to the
distance-$r$ dominating set problem is presented which
preserves the parameter~$k$ exactly. In that paper, the reduction
is used to prove $\mathrm{W}[2]$-hardness of 
the distance-$r$ dominating set problem. 

\begin{lemma}[Drange et al.~\cite{drange2016kernelization}]\label[lemma]{lemma:reduction}
Let $(U,\FFF,k)$ be an instance of set cover and let 
$r\in \N$. There exists a graph $G\in \mathcal{H}_{r}$
such that $(U,\FFF,k)$ is a positive instance of the
set cover problem if and only
if $(G,k)$ is a positive instance of the distance-$r$ dominating 
set problem.
\end{lemma}

Combining \Cref{lemma:fpt-approx}, \Cref{lemma:fpt-approx-lowerbound}, \Cref{lemma:somewheredense} and \Cref{lemma:reduction} now gives the following theorem. 

\begin{theorem}
If the Gap Exponential Time Hypothesis holds, then for every
monotone somewhere dense class of graphs $\CCC$ there is no $\alpha(k)$-approximate kernel for 
the distance-$r$ dominating set problem on $\CCC$ for
any function $\alpha\colon\N\rightarrow\N$. 
\end{theorem}

The same statement holds for the connected distance-$r$ 
dominating set
problem, as every graph that admits a distance-$r$ dominating
set of size $k$ also admits a connected distance-$r$ dominating
set of size at most $3k$. 

%% file: kernel_conclusion.tex
\section{Conclusion}

The study of computationally hard problems on restricted classes
of inputs is a very fruitful line of research in algorithmic graph structure
theory and in particular in parameterized complexity theory. This
research is based on the observation that many problems such as
\textsc{Dominating Set}, which are considered intractable in general,
can be solved efficiently on restricted graph classes. Of course it 
is a very desirable goal in this line of research to identify the most
general classes of graphs on which certain problems 
can be solved efficiently. In this work we were able to identify
the exact limit for the existence of lossy kernels for the connected
distance-$r$ dominating set problem. One interesting open question
is whether our polynomial bounds on the size of the 
lossy kernel can be improved to pseudo-linear bounds. The first
step to achieve this is to prove the existence of a 
$(k,r)$-domination core of pseudo-linear size on every nowhere
dense class of graphs, or to avoid the use of such cores in the
construction. 